\documentclass[aps,prd,reprint,amsmath,amssymb,nofootinbib,floatfix]{revtex4-2}

\usepackage[T1]{fontenc}
\usepackage{amsthm}
\usepackage{bm}
\usepackage{booktabs}
\usepackage{graphicx}
\usepackage{microtype}
\usepackage{xcolor}
\usepackage[colorlinks=true,allcolors=blue]{hyperref}
\graphicspath{{figures/}}

\newtheorem{theorem}{Theorem}
\newtheorem{corollary}[theorem]{Corollary}
\newtheorem{proposition}[theorem]{Proposition}
\newcommand{\pperp}{p_{\perp}}
\newcommand{\mWEC}{m_{\mathrm{WEC}}}
\newcommand{\Rone}{R_{1}}

\begin{document}

\title{Boundary Obstructions and Lapse Freedom in Static Spherical Hollow Cores}

\author{Nelson Bol\'ivar}
\email{nelson.e.bolivar@ucv.ve}
\affiliation{Astrum Drive Technologies, Dallas Pkwy Unit 120 B, Frisco, TX 75034, USA}
\affiliation{Departamento de F\'isica, Facultad de Ciencias, Universidad Central de Venezuela,
Av.\ Los Ilustres, Caracas 1041-A, Venezuela}

\author{Gabriel Abell\'an}
\affiliation{Astrum Drive Technologies, Dallas Pkwy Unit 120 B, Frisco, TX 75034, USA}
\affiliation{Departamento de F\'isica, Facultad de Ciencias, Universidad Central de Venezuela,
Av.\ Los Ilustres, Caracas 1041-A, Venezuela}

\author{Ivaylo Vasilev}
\affiliation{Astrum Drive Technologies, Dallas Pkwy Unit 120 B, Frisco, TX 75034, USA}

\date{August 14, 2026}

\begin{abstract}
We analyze a sharply delimited comparison problem for a static spherical hollow
core: an empty, flat cavity surrounded by a positive matter wall.  In the
unit-lapse, flat-slice radial Painlev\'e--Gullstrand (PG) class the Type-I source
obeys $p_r=-\rho$ and $p_\perp=-\rho-r\rho'/2$.  A regular nonnegative density
that rises out of the cavity must therefore violate the transverse null and
weak energy conditions.  We quantify this boundary obstruction by an exact
weighted onset budget and a depth--width bound, and show that the sharp limit
has the same negative tangential pressure as its symmetry-invariant
regularized Israel layer.  We then close the adjacent unit-lapse,
curved-slice route under a monotone areal-radius hypothesis.  Finally, when
only the lapse is released, an incomplete-beta family gives regular hollow
shells with flat cavities, Schwarzschild exteriors, no thin shells, and
NEC/WEC/SEC/DEC on an explicit compactness interval.  The family is
conditionally characterized in the minimum-degree and $p_r=0$ sectors and
admits a fixed-ADM-mass cavity redshift benchmark.  The result is a static
boundary theorem and construction, not a general hollow-shell existence
theorem, a transport result, or an experimental feasibility claim.
\end{abstract}

\maketitle

\section{Motivation and statement of the problem}

Hollow static configurations---matter walls enclosing an empty, locally flat
cavity---recur across general relativity: gravastars and their anisotropic
generalizations \cite{MazurMottola2004,VisserWiltshire2004,CattoenFaberVisser2005},
Einstein clusters \cite{Einstein1939,Florides1974}, and static Einstein--Vlasov
shells \cite{Rein1994,Rein1999Shell,AndreassonRein2007}.  The structural
question for any such object is what its boundary stresses permit under the
local energy conditions.  This paper establishes a restricted comparison
theorem for that question.  Within the two adjacent static extensions analyzed
here, varying the lapse is sufficient to obtain a regular positive-energy
hollow shell.

Positive-energy hollow and annular configurations are not new in static
spherical general relativity.  They occur in Einstein--Vlasov shell
constructions, gravastar models, thin-shell vacuum-bubble geometries, and
Einstein-cluster-type sources
\cite{Rein1999Shell,HorvatIlijic2007,RosaPiccarra2020,AcunaCardenas2024}.
Our contribution is different: we provide a boundary-complete comparison of
the restricted unit-lapse radial PG class and its two adjacent static
extensions.  The novelty lies in the structural relation between the onset
budget, the junction limit, the curvature-only closure, and the regular lapse
family, rather than in a new general matter model or a new existence theorem
for hollow shells.

The restricted class arises as follows.  Warp-drive research has repeatedly
moved between prescribed metrics and prescribed matter sources.  The original
Alcubierre and Nat\'ario constructions
\cite{Alcubierre1994,Natario2002} prescribe kinematics and infer a source; physical
warp-shell programs instead emphasize positive ADM mass, explicit matter, and
energy-condition diagnostics \cite{BobrickMartire2021,Fuchs2024}.  General
analyses show both that null-energy violation is generic in standard moving
warp fields \cite{Santiago2022} and that many apparently physical results depend
on severe metric and coordinate restrictions \cite{Barzegar2026}.

The authors' earlier piecewise construction used precisely the unit-lapse radial
PG class and showed how weak and null energy conditions could be retained at the
price of non-differentiable interfaces \cite{BolivarAbellanVasilev2025}.  Related
work with nontrivial lapse exhibited a broader anisotropic and heat-flux sector
\cite{AbellanBolivarVasilev2024}.  The present paper revisits the boundary between
those regimes.  Finite-resolution profile tests can hide a localized negative
margin, but they cannot change an exact sign identity.  The appropriate
question is therefore structural:

\begin{quote}
Can a regular density rise from a flat, empty core into a positive matter wall
while preserving NEC and WEC inside the unit-lapse, flat-slice, radial-shift
class?
\end{quote}

The answer is no in this restricted class.  We first quantify the obstruction
and the junction cost of its low-regularity limit.  We then show that curving
the spatial slices at unit lapse still does not help, and finally construct a
smooth positive-energy family after relaxing the lapse alone.  Our contribution
is this boundary-sensitive chain: the integrated cost of a smooth onset, its
Israel-layer counterpart in the sharp limit, the curvature-only closure at unit
lapse, and a regular lapse family with no surface layer.  Positive-energy
static shells are known; we do not claim a new general existence theorem for
them.  Rather, we make this boundary comparison and the lapse-only
construction explicit within one controlled family.

The remainder of the paper supplies two robustness checks rather than new
claims of matter-model priority.  We give the Einstein-cluster interpretation
of the exact family and show that its energy-condition margins survive an open
neighborhood with nonzero radial pressure.  A standard Einstein--Vlasov
constitutive companion and its numerical audit are retained in the
Supplemental Material, where they are explicitly treated as prior-art closure
and not as a new existence or stability theorem.  The onset budget and
boundary analysis complement recent source-first work locating
energy-condition failures at shell boundaries \cite{Le2026}.

\paragraph{Claim boundary.}
The paper proves statements only for the metric classes, regularity hypotheses,
and endpoint conditions stated below.  It does not classify arbitrary static
spherical matter, establish perturbative stability or formation, or identify
the PG shift with transport.  The novelty claim is the closed comparison
between (i) the unit-lapse onset obstruction and its regularized junction
measure, (ii) the unit-lapse curvature-only closure, and (iii) the explicit
lapse-only positive family.  Positive hollow sources, Einstein clusters, and
Einstein--Vlasov shells themselves are prior art.

We use ``warp-drive class'' only as a literature label: the configuration
studied here is static and spherical, its PG shift is a slicing velocity, and
our result concerns this restricted ansatz, not warp drives in general.

\section{Restricted geometry and exact source}

In geometric units $G=c=1$, consider
\begin{equation}
 ds^{2}=-dt^{2}+\left[dr-\beta(r)dt\right]^{2}+r^{2}d\Omega^{2}.
 \label{eq:metric}
\end{equation}
The lapse is unity, the spatial metric is Euclidean in spherical coordinates, and
the radial shift is curl-free.  Define
\begin{equation}
 w(r)\equiv r\beta(r)^{2}.
 \label{eq:wdef}
\end{equation}
A direct Einstein-tensor calculation in the orthonormal frame of the Eulerian
normal $n^{\mu}=(1,\beta,0,0)$ gives
\begin{align}
 8\pi\rho &= \frac{w'}{r^{2}}
 =\frac{\beta(2r\beta'+\beta)}{r^{2}},
 \label{eq:rho}\\
 p_{r} &= -\rho,
 \label{eq:pr}\\
 \pperp &= -\frac{w''}{16\pi r}
 =-\rho-\frac{r}{2}\rho',
 \label{eq:pperp}\\
 q_{r} &=0.
 \label{eq:flux}
\end{align}
The source is diagonal and Hawking--Ellis Type I wherever the classical tensor is
defined \cite{HawkingEllis1973}.  These pressure relations are not new.  In the
spherical vacuum-dark-fluid literature the conservation equations give the same
transverse-pressure identity, and WEC is
known to require a nonincreasing density \cite{Dymnikova1992,Bronnikov2012}.  The
same stress pattern $p_r=-\rho$ underlies gravastar interiors, where the
impossibility of supporting such walls with isotropic pressure is expressed as
a requirement of anisotropic stresses
\cite{MazurMottola2004,VisserWiltshire2004,CattoenFaberVisser2005};
the hollow-onset statements below sharpen that anisotropy requirement into
quantitative onset costs for the PG class.  We
specialize the known structure to a regular hollow PG onset and make its
integrated budget and low-regularity alternatives explicit.

For $\rho>0$, introduce the logarithmic slope
\begin{equation}
 s(r)\equiv-\frac{d\ln\rho}{d\ln r}=-\frac{r\rho'}{\rho}.
\end{equation}
Then $\pperp=(s-2)\rho/2$.  The independent principal energy-condition scalars are
\begin{align}
 \rho+p_{r}&=0,\nonumber\\
 \rho+\pperp&=-\frac{r}{2}\rho'=\frac{s}{2}\rho,
 \label{eq:necmargins}\\
 \rho+p_{r}+2\pperp&=(s-2)\rho,\nonumber\\
 \rho-\lvert\pperp\rvert&\ge0\quad\Longleftrightarrow\quad 0\le s\le4.\nonumber
\end{align}
Thus NEC and WEC require $\rho\ge0$ and $\rho'\le0$; all four standard pointwise
conditions require $\rho\ge0$ and $2\le s\le4$.  The latter slope window is useful
context.  The hollow-core statement below is a direct corollary of the previously
known WEC monotonicity condition, not a new general no-go theorem.

One bookkeeping correction is important.  If the WEC margin is defined as the
minimum principal inequality, then for $\rho>0$
\begin{equation}
 \mWEC
 =\min\{\rho,\rho+p_{r},\rho+\pperp\}
 =\min\left\{0,-\frac{r}{2}\rho'\right\}.
 \label{eq:mwec}
\end{equation}
It is not identically zero: it is saturated for a nonincreasing density and
negative wherever the density rises.  (The slope variable $s$ and
Eq.~\eqref{eq:mwec} are defined only where $\rho>0$; on the zero set of $\rho$
the margins are evaluated directly from Eq.~\eqref{eq:necmargins}.)

\section{Smooth hollow-core obstruction}

\begin{corollary}[Hollow-core onset obstruction]
\label{thm:onset}
Let $\Rone>0$ and let $\rho$ be continuous on $[\Rone,R_{2}]$ and differentiable
on $(\Rone,R_{2})$.  Suppose that
\begin{enumerate}
 \item $\rho(\Rone)=0$,
 \item $\rho(r)\ge0$ on $[\Rone,R_{2}]$, and
 \item $\rho(r_{\star})>0$ for at least one $r_{\star}\in(\Rone,R_{2}]$.
\end{enumerate}
Then the metric class~\eqref{eq:metric} violates the transverse NEC, and hence the
WEC, at some point in $(\Rone,r_{\star})$.
\end{corollary}

\begin{proof}
By the mean value theorem there exists $c\in(\Rone,r_{\star})$ such that
\begin{equation}
 \rho'(c)=\frac{\rho(r_{\star})-\rho(\Rone)}{r_{\star}-\Rone}>0.
\end{equation}
Equation~\eqref{eq:necmargins} then gives
$\rho(c)+\pperp(c)=-c\rho'(c)/2<0$.
\end{proof}

\begin{proposition}[Onset budget]
\label{prop:budget}
If $\rho\in C^{1}([\Rone,b])$ for some $b>\Rone$, then
\begin{equation}
 \mathcal{B}(\Rone,b)
 \equiv 2\int_{\Rone}^{b}\frac{\rho+\pperp}{r}\,dr
 =\rho(\Rone)-\rho(b).
 \label{eq:onsetbudget}
\end{equation}
In particular, if $\rho(\Rone)=0$ and $\rho(b)>0$, then
$\mathcal{B}(\Rone,b)=-\rho(b)<0$ and
\begin{equation}
 \min_{r\in[\Rone,b]}(\rho+\pperp)
 \le -\frac{\rho(b)}{2\ln(b/\Rone)}.
 \label{eq:depthbound}
\end{equation}
\end{proposition}

\begin{proof}
Substitution of Eq.~\eqref{eq:necmargins} into the integral gives
$\mathcal{B}=-\int_{\Rone}^{b}\rho'\,dr$, proving
Eq.~\eqref{eq:onsetbudget}.  Equation~\eqref{eq:depthbound} follows because the
minimum of $\rho+\pperp$ cannot exceed its weighted average with measure
$dr/r$.  For $b-\Rone\ll\Rone$, the bound scales as
$-\rho(b)\Rone/[2(b-\Rone)]$.
\end{proof}

\begin{corollary}[Regular flat-core shift]
If $\beta$ is sufficiently regular for the classical Einstein tensor,
$\beta=0$ on $r\le\Rone$, and the resulting $\rho$ is differentiable outside the
core, then a nontrivial nonnegative matter wall cannot satisfy WEC everywhere.
\end{corollary}

Indeed, a regular match has $\beta(\Rone)=0$ with bounded $\beta'$, so
$\beta\beta'\rightarrow0$ and Eq.~\eqref{eq:rho} gives $\rho(\Rone)=0$.
Corollary~\ref{thm:onset} applies as soon as the exterior density becomes positive.  Equivalently,
WEC forces $\rho$ to be nonincreasing; a nonnegative function cannot start at zero,
decrease, and later become positive.

Figure~\ref{fig:onset-obstruction} compares a $C^{2}$ Hermite-quintic
(``smootherstep'') onset followed by a decaying
tail with a filled monotone profile.  The hollow profile has positive energy
density everywhere in its wall, yet it necessarily develops a negative transverse
null margin while turning on.  The filled profile avoids the deficit because it
starts positive at the center and decreases outward.  The figure illustrates the
corollary; it is not numerical evidence for it.

\begin{figure}[t]
 \includegraphics[width=0.92\linewidth]{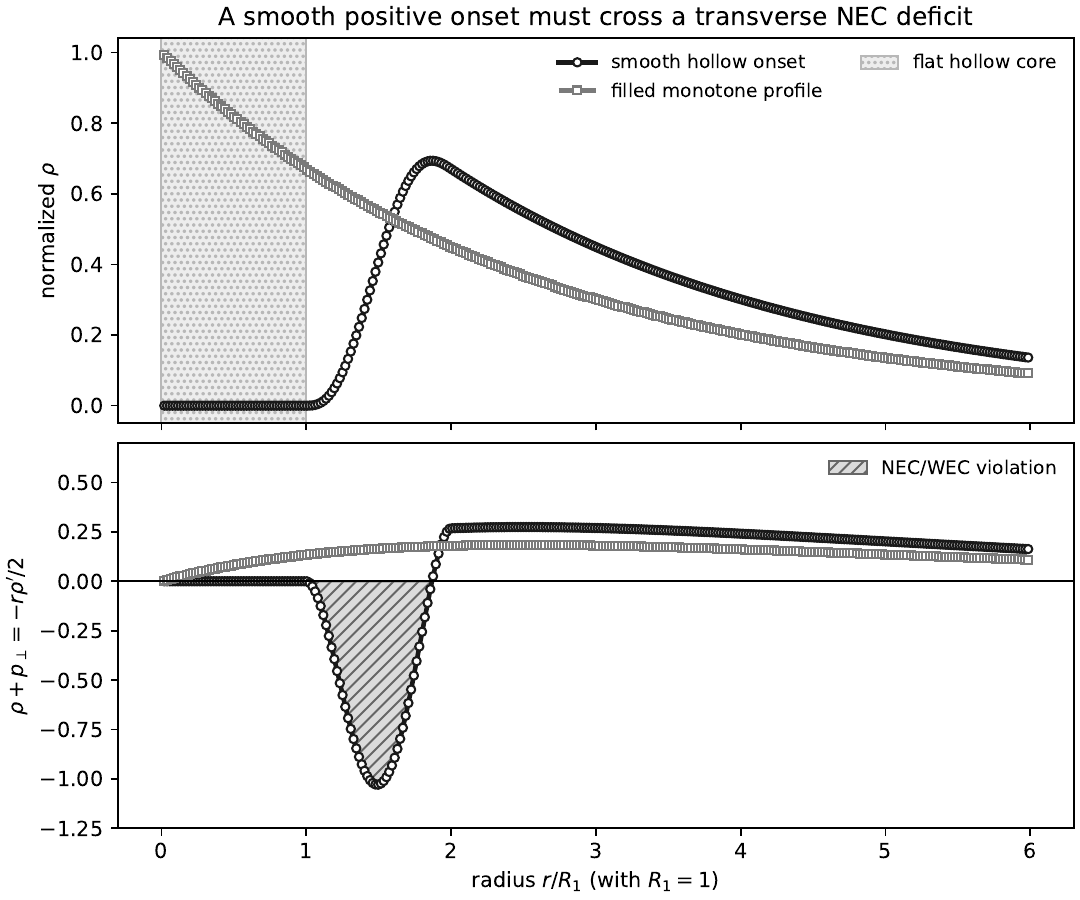}
 \caption{Normalized profiles.
 A smooth hollow source (solid) rises from zero and therefore has
 $\rho+\pperp<0$ in part of the onset.  A filled monotone source (dashed)
 satisfies the transverse NEC.  The shaded core is empty for the hollow
 profile.}
 \label{fig:onset-obstruction}
\end{figure}

\section{The boundary-regularity trilemma}

The corollary assumes a regular density.  Could one turn on the wall discontinuously?
The exact mass variable exposes the price.  Near $r=\Rone$, write
\begin{equation}
 \beta(r)\sim A(r-\Rone)^{n},\qquad A>0.
 \label{eq:poweronset}
\end{equation}
For $r>\Rone$,
\begin{equation}
 \rho(r)\sim\frac{nA^{2}}{4\pi\Rone}
 (r-\Rone)^{2n-1}.
 \label{eq:rhopower}
\end{equation}
The three possibilities are collected in Table~\ref{tab:trilemma}.

\begin{table}[t]
\caption{Onset taxonomy for shifts with power-law leading order
$\beta\sim A(r-\Rone)^{n}$.}
\label{tab:trilemma}
\footnotesize
\begin{ruledtabular}
\begin{tabular}{lll}
Exponent & Density onset & Consequence \\
\hline
$n>1/2$ & $\rho\to0$ and rises & NEC/WEC deficit \\
$n=1/2$ & finite positive step & $\beta'\sim(r-\Rone)^{-1/2}$ \\
$n<1/2$ & $\rho$ diverges & singular source \\
\end{tabular}
\end{ruledtabular}
\end{table}

A genuinely $C^{1}$ match to the flat core requires $n>1$, and therefore lies
unambiguously in the first row.  The square-root onset often used to generate a
finite positive wall density lies in the second row: $\beta$ is continuous but not
$C^{1}$.  It must not be described as a smooth junction.

The reduced equations first suggest a formal distributional diagnostic.  Used
alone, it would manipulate distributions inside nonlinear curvature expressions
outside the Geroch--Traschen class \cite{GerochTraschen1987}.  We therefore state
it here only to identify the candidate coefficient, and justify that coefficient
below as the weak limit of smooth geometries.  Let
\begin{equation}
 \rho=H(r-\Rone)\rho_{+}(r),\qquad \rho_{+}(\Rone)>0,
 \label{eq:step}
\end{equation}
where $H$ is the Heaviside distribution.  Then
\begin{equation}
 \rho'=H\rho_{+}'+\rho_{+}(\Rone)\delta(r-\Rone),
\end{equation}
and the reduced identity~\eqref{eq:necmargins} contains
\begin{equation}
 (\rho+\pperp)_{\mathrm{surface}}
 =-\frac{\Rone}{2}\rho_{+}(\Rone)\delta(r-\Rone).
 \label{eq:surfacenec}
\end{equation}
Equivalently, $w'$ jumps by $8\pi\Rone^{2}\rho_{+}(\Rone)$ and
$\pperp=-w''/(16\pi r)$ carries the same negative surface coefficient.

The same one-sided boundary data admit a rigorous thin-shell completion in
static diagonal coordinates.  In each open region with $|\beta|<1$, the change
$dT=dt+\beta(1-\beta^2)^{-1}dr$ gives
\begin{equation}
 ds^2=-f(r)dT^2+\frac{dr^2}{f(r)}+r^2d\Omega^2,
 \qquad f=1-\beta^2=1-\frac{w}{r}.
 \label{eq:diagonal-pg}
\end{equation}
Treating the two sufficiently regular diagonal sides as an Israel junction,
use one normal pointing from the cavity toward increasing $r$.  At
$r=\Rone$, where $f_-=f_+=1$, the surface tensor is
\begin{equation}
 \sigma=0,
 \qquad
 P_s=\frac{[f']}{16\pi}
 =-\frac{w'_+(\Rone)}{16\pi\Rone}
 =-\frac{\Rone}{2}\rho_+(\Rone)<0 .
 \label{eq:israel-onset}
\end{equation}

The low regularity of the limiting PG chart can be handled without applying
this time change directly to a distribution.  The appropriate statement is an
association theorem for smooth geometries.

\begin{proposition}[Regularized PG--Israel association]
\label{prop:pg-israel-association}
Let $w=0$ for $r<\Rone$ and $w=w_+(r)$ for $r>\Rone$, where $w_+$ is smooth
for $r>\Rone$, has a $C^2$ one-sided extension, $w_+(\Rone)=0$, and
$J\equiv w_+'(\Rone)>0$.  Assume $f=1-w/r$ has a positive lower bound in a
collar of the junction.  Let $g_\epsilon$ be any smooth, horizon-free PG family
of the form~\eqref{eq:metric}, with
$w_\epsilon=r\beta_\epsilon^2$, such that
\begin{equation}
 \begin{split}
 w_\epsilon&\longrightarrow w\quad\hbox{in }W^{1,1}_{\mathrm{loc}},
 \qquad \beta_\epsilon\longrightarrow\beta\quad\hbox{locally uniformly},\\
 &\sup_\epsilon \operatorname{TV}_K(w_\epsilon')<\infty
 \end{split}
 \label{eq:admissible-regularization}
\end{equation}
for every compact collar $K$.  Writing
$\boldsymbol G_\epsilon{}^\mu{}_{\nu}
\equiv G^\mu{}_{\nu}[g_\epsilon],dV_{g_\epsilon}$, these tensor-valued Radon
measures converge weakly.  Their singular part is independent of the admissible
regularization and is
\begin{align}
 \left(\boldsymbol G^\mu{}_{\nu}\right)_{\mathrm{sing}}
 &=8\pi P_s\,\Pi^\mu{}_{\nu}\,\delta_\Sigma,
 \label{eq:associated-einstein-measure}\\
 P_s&=-\frac{J}{16\pi\Rone}
 =-\frac{\Rone}{2}\rho_+(\Rone).\notag
\end{align}
where $\Pi^\mu{}_{\nu}$ projects onto the two angular directions and
$\delta_\Sigma$ is the invariant measure on $r=\Rone$.  Thus the associated
surface tensor has $\sigma=0$ and equals the Israel tensor in
Eq.~\eqref{eq:israel-onset}.
\end{proposition}

\begin{proof}
For every $\epsilon>0$, a direct calculation in the off-diagonal PG chart gives
the mixed components
\begin{align}
 G^t{}_t=G^r{}_r&=-\frac{w_\epsilon'}{r^2}, &
 G^\theta{}_\theta=G^\phi{}_\phi&=-\frac{w_\epsilon''}{2r},
 \label{eq:pg-mixed-einstein}\\
 G^\mu{}_{\nu}&=0\quad(\mu\ne\nu).\notag
\end{align}
Thus the potentially singular expressions combine into derivatives of the
invariant mass variable before the limit is taken; no off-diagonal or additional
nonlinear concentration remains.  The $W^{1,1}$ convergence makes
$\rho_\epsilon$ and $p_{r,\epsilon}=-\rho_\epsilon$ converge locally in
$L^1$, so neither can acquire a delta term.  The uniform variation bound and
the uniqueness of the distributional derivative give
\begin{equation}
 w_\epsilon''\,dr\stackrel{*}{\rightharpoonup}
 H(r-\Rone)w_+''\,dr+J\,\delta(r-\Rone).
 \label{eq:w-second-weak-limit}
\end{equation}
Equation~\eqref{eq:pperp} therefore gives the coefficient in
Eq.~\eqref{eq:associated-einstein-measure}.  The volume density is
$\sqrt{-g_\epsilon}=r^2\sin\theta$, while the angular projector is smooth and
independent of the PG time slicing.  Pairing with an arbitrary compactly
supported smooth test tensor consequently yields the displayed tensor measure;
the remaining PG components contain only bounded factors multiplying
$\rho_\epsilon$ and $p_{r,\epsilon}$ and add no singular part.  At the junction
$f(\Rone)=1$, so proper normal distance obeys $d\ell=dr$ and
$\delta_\Sigma$ has the same radial coefficient as $\delta(r-\Rone)$.

This argument is carried out in the PG chart itself and does not push a
distribution through the limiting time map.  For each fixed $\epsilon$, the
smooth diffeomorphism
$dT_\epsilon=dt+\beta_\epsilon f_\epsilon^{-1}dr$ gives exactly the diagonal
metric~\eqref{eq:diagonal-pg}.  Hence the two descriptions are identical before
the weak limit is taken, and the limiting angular measure is precisely the
Israel layer.  Finally, both $r$ and $w=2m_{\rm MS}$ are invariant scalars in
spherical symmetry.  Conditions~\eqref{eq:admissible-regularization} therefore
select the geometry rather than an arbitrary radial component smoothing, and
the limit depends only on $J=[w']$, not on the mollifier.

The admissible class is nonempty and can preserve the exact vacuum core.  Let
$\chi\in C^\infty(\mathbb R)$ be a flat step with $\chi(u)=0$ for $u\leq0$
and $\chi(u)=1$ for $u\geq1$.  On the exterior sign branch set
\begin{equation}
 \beta_\epsilon(r)=
 \chi\!\left(\frac{r-\Rone}{\epsilon}\right)
 \sqrt{\frac{w_+(r)}{r}},
 \label{eq:explicit-pg-regularization}
\end{equation}
and set $\beta_\epsilon=0$ in the cavity.  Flatness of $\chi$ removes the
square-root corner, while the family agrees with the target for
$r\geq\Rone+\epsilon$ and satisfies
Eq.~\eqref{eq:admissible-regularization}.  A further ordinary smoothing away
from the collar, if needed, does not alter the shell limit.
\end{proof}

Thus $\sigma+P_s<0$: this diagonal thin-shell completion violates its
intrinsic null and weak energy conditions.  Its coefficient agrees with the
surface term in Eq.~\eqref{eq:surfacenec}, since
$P_s\delta(\ell)=P_s\delta(r-\Rone)$ when $f(\Rone)=1$.  For
$\beta\sim A\sqrt{r-\Rone}$ one has
$\rho_+=A^2/(8\pi\Rone)$ and $P_s=-A^2/(16\pi)$.  In this limit the original
PG metric is only $C^{0,1/2}$ at the onset, while the time transformation above
is $C^1$ but not $C^{1,1}$.  Indeed,
$\beta'^2\sim J/[4\Rone(r-\Rone)]$, so the raw metric does not satisfy the
local $L^2$ connection regularity used in the standard distributional-curvature
framework \cite{GerochTraschen1987,LeFlochMardare2007}.  Proposition
\ref{prop:pg-israel-association} establishes the associated Einstein measure;
it does not claim that the unsmoothed PG metric has a classical distributional
Riemann tensor.  Defining the latter would require a nonlinear generalized-
function framework \cite{SteinbauerVickers2006}.  Reversing the normal reverses
both the jump convention and the displayed extrinsic curvatures, but not the
physical surface tensor
\cite{Israel1966,GerochTraschen1987,Racsko2022}.  The classical onset
alternatives are:
Corollary~\ref{thm:onset} covers every differentiable onset regardless of
leading order, Table~\ref{tab:trilemma} organizes the power-law leading
orders, and Eq.~\eqref{eq:israel-onset} supplies the corresponding diagonal
thin-shell benchmark.  Within the restricted class, a positive hollow onset is
therefore smooth but NEC-violating, finite but non-$C^1$ and associated through
Proposition~\ref{prop:pg-israel-association} with a negative junction pressure,
or singular.

\section{How the onset bound controls profile design}

One might try to evade the corollary by making the onset very thin.  The
integrated identity shows why this cannot work.  At fixed wall density,
narrowing the rising interval cannot remove the onset cost, because it cannot
change the sign identity $\rho+\pperp=-r\rho'/2$ or the integrated
budget~\eqref{eq:onsetbudget}.  At fixed wall density, the minimum transverse
margin therefore becomes at least as negative as
Eq.~\eqref{eq:depthbound}; in the discontinuous limit the same budget is
represented by the negative surface term in Eq.~\eqref{eq:israel-onset}.  A
finite-resolution plot may fail to display a localized deficit, but no choice
of sampling changes the sign identity or its integral.
Numerical profiles are useful here as illustrations of how the deficit is
distributed, while admissibility is settled by the exact identities above.  The
same distinction applies to the phrase ``positive energy'': a profile with
$\rho\ge0$ has nonnegative
energy density in the Type-I rest/Eulerian frame, which is not WEC unless the
principal sums are also nonnegative; recent irrotational constructions draw
the same distinction \cite{Rodal2026}.

\section{Curved slices at unit lapse do not evade the obstruction}
\label{sec:curvature-only}

One might try to preserve unit lapse but curve the spatial slices.  Within the
class of static spherical unit-lapse metrics whose areal radius is a global
monotone coordinate on the slice, the general radial ansatz is
\begin{equation}
 ds^2=-dt^2+B(r)^2[dr-\beta(r)dt]^2+r^2d\Omega^2,
 \qquad B(r)>0 .
 \label{eq:curved-slice-metric}
\end{equation}
Let $n^\mu=(1,\beta,0,0)$ be the future unit normal and
$e_{\hat r}^{\mu}=(0,B^{-1},0,0)$ the outward radial unit vector.  The two
future radial null vectors are $k_\pm^\mu=n^\mu\pm e_{\hat r}^{\mu}$.  Direct
contraction of the Einstein tensor gives
\begin{align}
 8\pi T_{\mu\nu}k_+^\mu k_+^\nu
 &=\frac{2B'}{rB^3}(1+B\beta)^2, \label{eq:null-plus-B}\\
 8\pi T_{\mu\nu}k_-^\mu k_-^\nu
 &=\frac{2B'}{rB^3}(1-B\beta)^2. \label{eq:null-minus-B}
\end{align}
Their sum removes every possible zero of an individual square:
\begin{equation}
 8\pi T_{\mu\nu}(k_+^\mu k_+^\nu+k_-^\mu k_-^\nu)
 =\frac{4B'}{rB^3}(1+B^2\beta^2).
 \label{eq:null-sum-B}
\end{equation}
Equations~\eqref{eq:null-plus-B}--\eqref{eq:null-minus-B} follow directly
from the radial and angular components of the Einstein tensor.

\begin{theorem}[Curvature-only closure in the monotone branch]
\label{thm:curvature-only}
Assume Eq.~\eqref{eq:curved-slice-metric} is regular for $r\geq\Rone$ with
$B\in C^1$, obeys
the radial NEC in both directions, has a flat inner boundary
$B(\Rone)=1$, and is asymptotically flat with
$\lim_{r\to\infty}B(r)=1$.  Then $B(r)\equiv1$ for all
$r\geq\Rone$.
\end{theorem}

\begin{proof}
Since $r>0$, $B>0$, and $1+B^2\beta^2>0$, the NEC and
Eq.~\eqref{eq:null-sum-B} imply $B'\geq0$.  Thus $B$ is nondecreasing.
A nondecreasing function beginning at one and tending to one at infinity is
identically one.
\end{proof}

Two scope remarks are required.  First, the theorem is independent of the
detailed shift profile: a compact spatial curvature bump must return to $B=1$
and therefore has a region with $B'<0$, where at least one radial null
contraction is negative.  Second, Eq.~\eqref{eq:curved-slice-metric} presumes
that the areal radius is a global monotone coordinate; static slices whose
sphere-area function is non-monotone (throat- or neck-type geometries, with
$R'=0$ somewhere) lie outside the ansatz, and the theorem does not decide
them.  This is a closure result for the stated unit-lapse,
monotone-areal-radius branch, not a classification of all static spherical
geometries.  Within monotone-areal-radius static spherical symmetry, spatial
curvature by itself is therefore not an escape from the unit-lapse hollow-core
obstruction.

\section{A regular positive-energy family from lapse freedom}
\label{sec:escape}

The preceding results isolate the next single freedom in this sequence: the
lapse.  We now give a family rather than a single favorable profile.  In
diagonal static
coordinates let
\begin{equation}
 ds^2=-e^{2\Phi(r)}d\tau^2+\frac{dr^2}{1-2m(r)/r}+r^2d\Omega^2 .
 \label{eq:positive-shell-metric}
\end{equation}
Fix a wall aspect ratio $\eta>1$ and an integer $n\geq1$.  With
$x=(r/L-1)/(\eta-1)$ define the regularized-beta profile
\begin{align}
 F_n(x)&=I_x(n+1,n+1)
 =\frac{\int_0^x u^n(1-u)^n du}{\mathrm{B}(n+1,n+1)},
 \label{eq:beta-family}\\
 F_n'(x)&=\frac{x^n(1-x)^n}{\mathrm{B}(n+1,n+1)} .
 \label{eq:beta-family-derivative}
\end{align}
The function rises monotonically from zero to one, while its first $n$
derivatives vanish at both endpoints.  It therefore turns on the mass across
the wall with controlled endpoint regularity.  The following proposition
states the precise, limited sense in which this familiar profile is the
simplest choice.

\begin{proposition}[Minimal-degree profile characterization]
\label{prop:profile-uniqueness}
Among polynomials $P$ of degree at most $2n+1$, the endpoint data
\begin{align}
 P(0)&=0,\qquad P(1)=1,\nonumber\\
 P^{(k)}(0)&=P^{(k)}(1)=0\qquad(1\leq k\leq n)
 \label{eq:profile-jets}
\end{align}
have the unique solution $P=F_n$.
\end{proposition}

\begin{proof}
The derivative $P'$ has zeros of multiplicity at least $n$ at both endpoints
and degree at most $2n$.  Hence
$P'(x)=C x^n(1-x)^n$.  The condition $P(1)-P(0)=1$ fixes
$C=\mathrm{B}(n+1,n+1)^{-1}$, and integration gives
Eq.~\eqref{eq:beta-family}.  Equivalently, the difference of two solutions has
zeros of multiplicity $n+1$ at both endpoints but degree at most $2n+1$, so it
vanishes identically.  The data~\eqref{eq:profile-jets} are classical
two-point Hermite conditions (the ``smoothstep'' family); the proposition is
standard interpolation theory restated to fix in what sense $F_n$ is
selected.  It is a minimum-degree statement; higher-degree and
nonpolynomial profiles are not excluded.
\end{proof}

Then set
\begin{equation}
 m(r)=
 \begin{cases}
 0, & 0\leq r\leq L,\\
 M F_n(x), & L<r<\eta L,\\
 M, & r\geq\eta L,
 \end{cases}
 \label{eq:family-mass-lapse}
\end{equation}
and normalize the lapse by
\begin{equation}
 \Phi'=\frac{m}{r(r-2m)},\qquad \Phi(\infty)=0 .
 \label{eq:family-lapse}
\end{equation}
The special case $n=4$, $\eta=3$ is the degree-nine profile
$F_4(x)=126x^5-420x^6+540x^7-315x^8+70x^9$ used below for numerical
 reference values.

\begin{proposition}[Conditional uniqueness in the $p_r=0$ sector]
\label{prop:przero-uniqueness}
Let a continuous, horizon-free mass function $m(r)$ be prescribed, with a flat
cavity, a constant exterior value $M$, and enough differentiability in the
matter region to define $\rho$.  Among metrics of the form
Eq.~\eqref{eq:positive-shell-metric}, normalized by $\Phi(\infty)=0$, imposing
$p_r=0$ determines $\Phi$ uniquely by Eq.~\eqref{eq:family-lapse}.  Where the
classical source is defined, its remaining components are necessarily
\begin{equation}
 \rho=\frac{m'}{4\pi r^2},\qquad
 p_\perp=\rho\,\frac{m}{2(r-2m)} .
 \label{eq:conditional-unique-source}
\end{equation}
\end{proposition}

\begin{proof}
The radial Einstein equation is
\begin{equation}
 8\pi p_r=-\frac{2m}{r^3}
 +\frac{2}{r}\left(1-\frac{2m}{r}\right)\Phi'.
 \label{eq:radial-einstein-general}
\end{equation}
Setting $p_r=0$ gives Eq.~\eqref{eq:family-lapse}; normalization at infinity
removes its single additive constant.  The temporal Einstein equation gives
$\rho=m'/(4\pi r^2)$, and $\nabla_\mu T^{\mu r}=0$ then fixes $p_\perp$ as in
Eq.~\eqref{eq:conditional-unique-source}.  This characterizes the prescribed-%
$m$, zero-radial-pressure sector only.  If $p_r(r)$ is free, it is an additional
source function and uniqueness is lost unless it or a constitutive law is also
specified.
\end{proof}

These equations also make the local static-equilibrium balance transparent.  In the
unit-lapse PG sector the geometry enforces $p_r=-\rho$: the radial NEC is
saturated, while the transverse combination is locked to
$\rho+p_\perp=-r\rho'/2$.  A rising onset therefore has the wrong sign.  Here
the redshift gradient instead enters the acceleration of the static observers
and the anisotropic force balance.  For $p_r=0$, stress-energy conservation
reads
\begin{equation}
 \frac{2p_\perp}{r}=\rho\Phi',
 \label{eq:lapse-support-balance}
\end{equation}
and Eq.~\eqref{eq:family-lapse} has $\Phi'>0$ wherever $m>0$ in the
horizon-free branch.  Thus $p_\perp\geq0$ for $\rho,m\geq0$, with strict
positivity where both are nonzero, without requiring a sign condition on
$\rho'$.  This is the equilibrium relation that replaces the unit-lapse sign
lock; by itself it neither identifies a microscopic support mechanism nor
establishes dynamical stability.

\begin{theorem}[Regular all-energy-condition hollow-shell family]
\label{thm:positive-shell}
Let $y=r/L$, $\mu=M/L$, and
\begin{equation}
 q_{n,\eta}=\max_{1\leq y\leq\eta}
 \frac{F_n[(y-1)/(\eta-1)]}{y} .
 \label{eq:q-family}
\end{equation}
Equations~\eqref{eq:positive-shell-metric}--\eqref{eq:family-lapse}
define a nontrivial, asymptotically flat, horizon-free $C^n$ spacetime with a
flat cavity, an exact Schwarzschild exterior, and no thin shell whenever
\begin{equation}
 0<\mu\leq\mu_{\rm DEC}(n,\eta)
 \equiv\frac{0.4}{q_{n,\eta}} .
 \label{eq:mu-family-interval}
\end{equation}
Its Type-I source satisfies NEC, WEC, SEC, and DEC everywhere.  DEC is
saturated at the upper endpoint; replacing $\leq$ by $<$ gives a strict DEC
margin.  For $n=4$, $\eta=3$,
$q_{4,3}=0.380436713748\ldots$ and
$\mu_{\rm DEC}=1.051423234257\ldots$.
\end{theorem}

\begin{proof}
The field equations reduce exactly to
\begin{equation}
 \rho=\frac{m'}{4\pi r^2},\qquad p_r=0,\qquad
 p_\perp=\rho\,\frac{m}{2(r-2m)} .
 \label{eq:positive-shell-source}
\end{equation}
Equation~\eqref{eq:beta-family-derivative} gives $\rho\geq0$.  Hence NEC,
WEC, and SEC hold in the horizon-free branch.  The only nontrivial DEC
inequality is $p_\perp\leq\rho$, or
\begin{equation}
 \frac{2m}{r}\leq\frac45 .
 \label{eq:dec-compactness}
\end{equation}
Since
\begin{equation}
 C_\star\equiv\max_r\frac{2m(r)}r=2\mu q_{n,\eta},
 \label{eq:cstar}
\end{equation}
Eq.~\eqref{eq:mu-family-interval} is necessary and sufficient for DEC within
this nontrivial static family and also gives $C_\star\leq0.8<1$.  In
particular $1-2m/r\geq0.2$ everywhere, so the metric
\eqref{eq:positive-shell-metric} is regular and horizon-free and the
expression for $p_\perp$ in Eq.~\eqref{eq:positive-shell-source} is finite.
We do not
assert that DEC by itself excludes horizons outside this branch.  Derivatives
of $F_n$ through order $n$ vanish at both endpoints, whereas the derivative of
order $n+1$ does not.  Thus $m$ and $g_{rr}$ are globally $C^n$ but not in
general $C^{n+1}$; the lapse equation makes $\Phi$ and $g_{\tau\tau}$ one order
smoother.  In the cavity $m=0$ and $\Phi$ is constant.  In the exterior,
$e^{2\Phi}=1-2M/r$, establishing the Schwarzschild match and ADM mass $M$.
\end{proof}

\noindent
Theorem~\ref{thm:positive-shell} is the central constructive result.  The
subsections that follow interpret its source, test its robustness when
$p_r=0$ is relaxed, and connect the geometry to a microscopic matter model
and an invariant weak-field diagnostic.  These extensions clarify the
physical scope of the family without changing the assumptions of the theorem.

\subsection{Kinetic interpretation and its causal domain}

The $p_r=0$ stress pattern belongs to the Einstein-cluster/Florides family
\cite{Einstein1939,Florides1974}.  In its ideal kinetic interpretation, massive
particles occupy circular orbits whose orientations and senses are distributed
spherically, canceling the net angular momentum.  If $v_{\rm tan}$ is the speed
measured by static orthonormal observers, angular averaging gives
$p_\perp=\rho v_{\rm tan}^2/2$.  Equation~\eqref{eq:positive-shell-source}
therefore fixes
\begin{equation}
 v_{\rm tan}^2=\frac{2p_\perp}{\rho}
 =r\Phi'=\frac{m}{r-2m}.
 \label{eq:cluster-speed}
\end{equation}
Massive subluminal constituents require
\begin{equation}
 \max_r\frac{2m}{r}<\frac23,\qquad
 0<\mu<\mu_{\rm kin}(n,\eta)
 \equiv\frac{1}{3q_{n,\eta}}.
 \label{eq:kinetic-compactness}
\end{equation}
This is a stricter material-realization domain than the abstract DEC endpoint
$2m/r\leq4/5$, but it remains nonempty for every $n$ and $\eta$.  It includes
the weak-field benchmark by an enormous margin.  We claim neither a new
microscopic mechanism nor that an ordinary solid shell has $p_r=0$: collisions,
confinement fields, supports, and perturbations generally add radial stresses
and their own stress-energy.  Equation~\eqref{eq:cluster-speed} supplies a
causal source interpretation, not a stability theorem or an apparatus model;
for numerical evidence on the stability of related static Einstein--Vlasov
states see Refs.~\cite{ReinRendall2000,AndreassonRein2007}.

\subsection{An open radial-pressure neighborhood}

The exactly soluble source is not isolated in stress space.  Write
$\mathcal F(y)=F_n[(y-1)/(\eta-1)]$, $x=(y-1)/(\eta-1)$, and
$s(x)=x(1-x)$.  Keep $m$ and hence $\rho$ fixed, but prescribe
\begin{equation}
 p_r^{(\epsilon)}=\epsilon\rho s,\qquad
 \Phi_\epsilon'=
 \frac{m+4\pi r^3p_r^{(\epsilon)}}{r(r-2m)},\qquad
 \Phi_\epsilon(\infty)=0 .
 \label{eq:radial-pressure-family}
\end{equation}
The conservation equation then determines the exact tangential stress,
\begin{equation}
 p_\perp^{(\epsilon)}=p_r^{(\epsilon)}
 +\frac r2\left[(p_r^{(\epsilon)})'
 +(\rho+p_r^{(\epsilon)})\Phi_\epsilon'\right].
 \label{eq:radial-pressure-tangential}
\end{equation}

\begin{theorem}[Robustness to nonzero radial pressure]
\label{thm:radial-pressure-robustness}
For every $n\geq1$, $\eta>1$, and strict interior point
$0<\mu<0.4/q_{n,\eta}$, there is an $\epsilon_\star>0$ such that every
$|\epsilon|<\epsilon_\star$ in
Eqs.~\eqref{eq:radial-pressure-family}--\eqref{eq:radial-pressure-tangential}
defines a horizon-free, junction-free hollow shell satisfying NEC, WEC, SEC,
and DEC everywhere.  Every nonzero $\epsilon$ gives $p_r\ne0$ in the open
matter wall.  The cavity and exterior remain respectively flat and exactly
Schwarzschild, with the same ADM mass.
\end{theorem}

\begin{proof}
It is useful to display a constructive uniform bound.  A prime on $\mathcal F$
below denotes $d/dy$.  Define the continuous wall functions
\begin{align}
 h_0&=\frac{\mu\mathcal F}{y-2\mu\mathcal F},&
 h_1&=\frac{\mu y\mathcal F's}{y-2\mu\mathcal F},\nonumber\\
 D&=\frac{y(n+1)(1-2x)}{\eta-1}-2s,\nonumber\\
 U&=s+\frac D2+\frac{s h_0+h_1}{2},&
 V&=\frac{s h_1}{2}.
 \label{eq:robustness-functions}
\end{align}
The endpoint limits are finite because the extra factor $s$ cancels the
logarithmic derivative of $\rho$; hence $h_0$, $h_1$, $U$, and $V$ extend
continuously to the closed wall $1\leq y\leq\eta$ and the maxima used below
are attained.  In the matter wall the pressure ratios are
exactly
\begin{equation}
 a_\epsilon\equiv\frac{p_r^{(\epsilon)}}\rho=\epsilon s,\qquad
 b_\epsilon\equiv\frac{p_\perp^{(\epsilon)}}\rho
 =\frac{h_0}{2}+\epsilon U+\epsilon^2V .
 \label{eq:robustness-ratios}
\end{equation}
Let $b_{\max}=\max(h_0/2)<1$,
$\mathcal A=\max|U|$, and $\mathcal B=\max V$.  The strict inequality follows
from the assumed interior DEC margin.  Any $e>0$ satisfying
\begin{align}
 e&<4,&
 \mathcal A e+\mathcal B e^2&<1-b_{\max},\nonumber\\
 \frac e4+2(\mathcal A e+\mathcal B e^2)&<1
 \label{eq:epsilon-witness}
\end{align}
is a valid lower witness for $\epsilon_\star$.  Indeed, $0\leq s\leq1/4$
gives $|a_\epsilon|<1$; the second inequality gives
$|b_\epsilon|<1$; and the last gives
$1+a_\epsilon+2b_\epsilon>0$.  Together with $\rho\geq0$, these are DEC and
SEC, and imply WEC and NEC.  Such an $e$ exists because all three right-hand
margins are positive.  The witness depends on $(n,\eta,\mu)$ through
$b_{\max}$, $\mathcal A$, and $\mathcal B$; as $\mu\to0.4/q_{n,\eta}$ the
margin $1-b_{\max}$ closes and $\epsilon_\star\to0$.  The mass function is
unchanged, so the horizon margin is unchanged.  Finally, $p_r^{(\epsilon)}$ and its first derivative vanish at
both interfaces for $n\geq1$; Eq.~\eqref{eq:radial-pressure-family} therefore
has the same bilateral lapse data used in the junction proof below.
\end{proof}

The geometric contribution is thus not a claim that zero radial pressure is
the only possible matter.  Propositions~\ref{prop:profile-uniqueness} and
\ref{prop:przero-uniqueness} identify what is unique under explicit minimal
assumptions, while Theorem~\ref{thm:radial-pressure-robustness} shows that the
positive-energy conclusion persists after one of those assumptions is relaxed.

The constitutive Einstein--Vlasov calculation is supplied as Supplemental
Material for readers who want a matter-level realization with nonzero radial
pressure.  It is a reproducibility and interpretation aid, not an additional
existence or stability result on which the main theorem depends.

\subsection{Darmois--Israel junction conditions}

Absence of a thin shell follows directly, rather than from visual smoothness.
For $f=1-2m/r$ and one normal $n^\mu=\sqrt f\,\delta^\mu_r$, the nonzero
extrinsic-curvature components of an $r={\rm const}$ worldtube are
\begin{align}
 K_{\tau\tau}&=-e^{2\Phi}\sqrt f\,\Phi',&
 K_{\theta\theta}&=r\sqrt f,\nonumber\\
 K_{\phi\phi}&=\sin^2\theta\,K_{\theta\theta}.&&
 \label{eq:family-extrinsic-curvature}
\end{align}
At $r=L$, both sides have $m=0$ and $\Phi'=0$; at $r=\eta L$, both sides
have $m=M$ and the same Schwarzschild value of $\Phi'$.  The induced metric
and every component in Eq.~\eqref{eq:family-extrinsic-curvature} therefore
match exactly: $[h_{ab}]=[K_{ab}]=0$ and the Israel tensor $S_{ab}$ vanishes.
Continuity of $m$ and the bilateral lapse closure are already sufficient for
this result; continuity of $m'$ is not required, so absence of an Israel layer
must not be conflated with continuity of the classical stress.  Within the
family stated here, $n=1$ is the minimum that also makes
$F_n'(0)=F_n'(1)=0$ and hence makes the stresses continuous and zero at both
interfaces.  Larger $n$ controls higher regularity but is not being used to
conceal a junction layer.  Together with Eq.~\eqref{eq:israel-onset}, this establishes both sides of
the boundary comparison: the positive family has no surface matter, whereas
the sharp unit-lapse PG onset has an unavoidable negative tangential layer.

\subsection{Invariant core depth and fixed-ADM comparison}

The constant cavity lapse is relationally observable once time is normalized
at infinity.  Write
$\mathcal F_{n,\eta}(y)=F_n[(y-1)/(\eta-1)]$.  Direct integration gives the
exact dimensionless map
\begin{align}
 \Phi_c(\mu;n,\eta)
 &=\frac12\ln\!\left(1-\frac{2\mu}{\eta}\right) \nonumber\\
 &\quad-\int_1^\eta
 \frac{\mu\mathcal F_{n,\eta}(y)}
 {y[y-2\mu\mathcal F_{n,\eta}(y)]}\,dy,
 \label{eq:exact-core-depth}\\
 -\Phi_c&=\kappa_{n,\eta}\mu+O(\mu^2),
 \label{eq:weak-core-depth}\\
 \kappa_{n,\eta}&=\frac1\eta+
 \int_1^\eta\frac{\mathcal F_{n,\eta}(y)}{y^2}\,dy .
 \label{eq:kappa-family}
\end{align}
For the $n=4$, $\eta=3$ reference,
\begin{equation}
 \kappa_{4,3}=-\frac{6975}{64}+\frac{25515}{256}\ln3
 =0.512080255339\ldots .
 \label{eq:kappa-s9}
\end{equation}
This replaces an arbitrary strong-field reference point by a continuous
weak-to-strong benchmark.

More importantly, the exterior field can be held fixed.  Compare two
homologous shells with the same ADM mass and $L_a<L_b$, keeping $n$ and $\eta$
fixed.  Restoring units with $M_g=GM_{\rm ADM}/c^2$,
\begin{align}
 \ln\frac{\nu_c(L_b)}{\nu_c(L_a)}
 &=\Phi_c(M_g/L_b)-\Phi_c(M_g/L_a) \nonumber\\
 &=\kappa_{n,\eta}\frac{GM_{\rm ADM}}{c^2}
 \left(\frac1{L_a}-\frac1{L_b}\right)+O(M_g^2/L^2).
 \label{eq:fixed-adm-rate}
\end{align}
Because the Schwarzschild exterior is identical in both configurations, this
is a geometry-sensitive fixed-ADM null test, not merely detection of an
ordinary change in total mass.  By Birkhoff's theorem no exterior measurement
distinguishes the two configurations; the entire signal is the cavity
potential, whose profile dependence enters through the order-unity
coefficient $\kappa_{n,\eta}$ (a thin shell at radius $L$ has $\kappa=1$ in
the same normalization).  The comparison detects an internal redistribution
of gravitating matter at fixed exterior field.  It is not by itself a unique
signature of hollow topology; that interpretation requires the source and
geometry controls stated above.  An ideal stationary matter-wave branch of
particle mass $m_a$ and dwell time $T$ has the corresponding leading proper-time
term
\begin{equation}
 \Delta\varphi_{\rm dwell}=\kappa_{n,\eta}
 \frac{G M_{\rm ADM}m_aT}{\hbar}
 \left(\frac1{L_a}-\frac1{L_b}\right).
 \label{eq:ideal-phase-benchmark}
\end{equation}
Equation~\eqref{eq:ideal-phase-benchmark} is a covariant phase contribution,
not a complete interferometer prediction.  It is included only as a relational
interpretation of the exact geometry; instrument thresholds and source
engineering are outside the paper.

The fixed-ADM comparison is therefore a geometric benchmark, not a claim of
experimental sensitivity or a microscopic scaling limit.  The corresponding
engineering extrapolations remain in the separate protocol and roadmap.

\subsection{Flat-slice representation}

The construction relaxes only the lapse, not the possibility of flat spatial
slices.  Write
\begin{equation}
 A(r)=e^{\Phi(r)},\qquad
 C(r)=\left(1-\frac{2m(r)}r\right)^{-1/2},
\end{equation}
and introduce a new time coordinate through
\begin{equation}
 \tau=t+h(r),\qquad
 h'(r)=\frac{\sqrt{C(r)^2-1}}{A(r)} .
 \label{eq:reslicing}
\end{equation}
Then Eq.~\eqref{eq:positive-shell-metric} becomes
\begin{equation}
 ds^2=-\alpha(r)^2dt^2+[dr-\beta(r)dt]^2+r^2d\Omega^2,
 \label{eq:flat-slice-escape}
\end{equation}
with
\begin{equation}
 \alpha=AC,\qquad
 \beta=A\sqrt{C^2-1}.
 \label{eq:escape-adm-fields}
\end{equation}
The spatial metric on $t=\mathrm{const}$ is exactly Euclidean.  In the cavity
$C=1$, so $\beta=0$ and $\alpha=e^{\Phi_c}$ is constant.  In the exterior,
 $A=C^{-1}$, hence $\alpha=1$ and $\beta=\sqrt{2M/r}$, the Schwarzschild PG
 slicing.  The nontrivial lapse is confined to the distinction between the
 asymptotic time normalization and the cavity.  Equations
\eqref{eq:flat-slice-escape}--\eqref{eq:escape-adm-fields} are therefore an
explicit counterexample to the idea that unit lapse is an inessential gauge
 restriction, obtained by relaxing only that assumption.  It does not contradict
 Theorem~\ref{thm:curvature-only}: its flat-slice representation has $B\equiv1$.
Although the geometry now has Euclidean spatial slices, the normalized norm of
the static Killing field---and hence the cavity-to-infinity clock ratio
$e^{\Phi_c}$ of Eq.~\eqref{eq:exact-core-depth}---is unchanged.  This factor
coincides with the flat-slice ADM lapse $\alpha=e^\Phi C$ in the cavity because
$C=1$ there; it should not be identified with $\alpha(r)$ throughout the
re-sliced spacetime.  Flat spatial geometry therefore does not remove the
physical redshift that labels the family.

\section{Relation to earlier work and scope}
\label{sec:scope}

The present paper continues the authors' earlier spherical-source program, but
changes the question being answered.  The 2023 study mapped numerical spherical
solutions with partial energy-condition control \cite{AbellanBolivarVasilev2023}.
The 2024 work explored nontrivial lapse, anisotropy, and heat flow
\cite{AbellanBolivarVasilev2024}.  The 2025 piecewise unit-lapse PG construction
obtained nonnegative bulk margins by allowing non-differentiable interfaces
\cite{BolivarAbellanVasilev2025}.  Here the interface itself is the object of
 study.  A sharp positive onset carries the negative Israel pressure
 \eqref{eq:israel-onset}, while any regular onset has the finite integrated
 budget of Eq.~\eqref{eq:onsetbudget}.  Thus the earlier bulk construction and
 its interior margins are not discarded; the present result identifies the
 boundary price that must be included when those interfaces are read as physical
 junctions.  Equation~\eqref{eq:israel-onset} is a precise boundary benchmark
 for the 2025 construction: its nonnegative bulk margins hold away from the
 interface, whereas a pointwise bulk test does not assess the associated
 junction layer.  Proposition~\ref{prop:pg-israel-association} now proves that
 every admissible smooth PG regularization has the same singular Einstein
 measure as the diagonal Israel completion.  The square-root limiting metric
 itself remains below the Geroch--Traschen regularity class, so this association
 must not be restated as existence of its classical distributional Riemann
 tensor.

\begin{table}[t]
\caption{Progression from the authors' earlier spherical models to the present
result.  ``All EC'' denotes NEC, WEC, SEC, and DEC.}
\label{tab:own-work-comparison}
\footnotesize
\begin{tabular}{@{}p{0.12\columnwidth}p{0.40\columnwidth}p{0.40\columnwidth}@{}}
\toprule
Work & Sector and EC result & Boundary control and distinct role \\
\midrule
2023 \cite{AbellanBolivarVasilev2023} & Numerical spherical fluids; partial EC control & Numerical exploration; no boundary theorem \\
2024 \cite{AbellanBolivarVasilev2024} & Nontrivial lapse, anisotropy, heat flow; broader positive-energy sectors & Opens source freedom; no minimal exact hollow junction theorem \\
2025 \cite{BolivarAbellanVasilev2025} & Unit-lapse radial PG; nonnegative bulk WEC/NEC margins & Non-$C^1$ onset; surface cost not included \\
Present & Unit-lapse obstruction, curvature-only closure, and regular $F_n$ lapse family; open $p_r\ne0$ neighborhood & Regular family has $S_{ab}=0$; the sharp PG onset is associated with the negative Israel layer; fixed-ADM core-depth comparison \\
\bottomrule
\end{tabular}
\end{table}

Several ingredients are standard and provide useful physical orientation.  The
pressure identity $p_r=-\rho$, $p_\perp=-\rho-r\rho'/2$ and its WEC
monotonicity are established in the spherical vacuum-dark-fluid literature
\cite{Dymnikova1992,Bronnikov2012}.  Einstein-cluster matter, gravastar-type
hollow interiors \cite{MazurMottola2004,VisserWiltshire2004}, and the
anisotropy requirement of Ref.~\cite{CattoenFaberVisser2005} all point to the
same lesson: a hollow wall is controlled by its boundary stresses, not only by
its bulk density.  Static Einstein--Vlasov shells and thin-shell vacuum-bubble
models already show that regular hollow sources satisfying pointwise energy
conditions can exist \cite{Rein1994,Rein1999Shell,AndreassonRein2007,
HorvatIlijic2007,RosaPiccarra2020,AcunaCardenas2024}, while broader
warp-drive classifications, source-first studies, and recent Alcubierre--
Minkowski junction analyses treat larger geometric classes
\cite{Barzegar2026,Le2026,SantosPereira2026}.  Our contribution within static
spherical symmetry is a linked package: a quantified PG-class onset cost and
its regularized Israel counterpart, a curvature-only closure result at unit
lapse, and an explicit positive family with regular junctions obtained by
relaxing the lapse.  For compactness context, the DEC endpoint $2m/r\leq4/5$ is the exact
pointwise DEC boundary of the $p_r=0$ sector, by
Eq.~\eqref{eq:conditional-unique-source}; it is not an instance of the sharp
bound $2m/r\leq8/9$ of
Refs.~\cite{Andreasson2008,KarageorgisStalker2008}, whose hypothesis
$p_r+2p_\perp\leq\rho$ fails at DEC saturation, where $p_r+2p_\perp=2\rho$.
Within that hypothesis class our source obeys the stricter causal bound
$2m/r<2/3$ of Eq.~\eqref{eq:kinetic-compactness}, since
$2p_\perp\leq\rho$ is equivalent to $v_{\rm tan}^2\leq1$ there.

All results remain restricted to static spherical symmetry.  The kinetic
Einstein cluster is an ideal cold limit, and the open radial-pressure theorem
shows only continuity in stress space; it is not a material support or
stability result.  The constitutive Einstein--Vlasov calculation is supplied
separately as a prior-art microscopic closure, not as part of the present
novelty claim.  Material supports, dynamical formation, perturbative stability,
and nonspherical deformations remain open.  Finally, the shift in
Eq.~\eqref{eq:flat-slice-escape} is a slicing field rather than a transport
velocity; dynamical transport is a distinct question.

\section{Conclusion}

The physical lesson is simple.  In the unit-lapse, flat-slice PG class, a
positive wall cannot rise out of an empty core without paying a transverse
null-energy cost.  Smoothing the onset spreads that cost through the wall;
admissible smooth PG regularizations of the sharp limit converge to the same
negative Israel surface pressure $P_s=-R\rho_+/2$.  The raw square-root chart
remains below the classical distributional-curvature threshold, but its
Einstein measure is unambiguous within the stated regularization class.
Allowing radial spatial curvature while keeping unit lapse
also fails, at least when the areal radius is monotone: the radial NEC makes
$B$ monotone, and the flat endpoint conditions force $B\equiv1$.

Within the static spherical comparison sequence considered here, lapse freedom
is sufficient, whereas radial spatial curvature alone is not.
Once the lapse is allowed to vary, the incomplete-beta family gives regular
hollow shells with Minkowski
cavities, Schwarzschild exteriors, no thin shell, and NEC, WEC, SEC, and DEC on
the analytic interval $0<\mu\leq0.4/q_{n,\eta}$.  The same spacetimes can be
written with Euclidean spatial slices after re-slicing, so the construction does
not rely on hidden spatial curvature.  What changes physically is the redshift
between the cavity and infinity.  The minimum-degree profile characterization
and the $p_r=0$ lapse reconstruction explain why this particular family is
natural, without making a uniqueness claim over all possible anisotropic matter.

The source extensions show that this result is not tied to an isolated stress
tensor.  The Einstein-cluster interpretation is causal on the stricter branch
$2m/r<2/3$, and the open $p_r\ne0$ neighborhood preserves the energy
conditions.  The invariant core-depth map then gives the family a direct
weak-field meaning: at fixed ADM mass the exterior Schwarzschild field is
unchanged while the cavity clock rate varies according to
Eq.~\eqref{eq:fixed-adm-rate}.  The separate engineering roadmap develops the
associated density--signal tradeoff; no laboratory sensitivity is inferred
from the present equation.

The paper thus establishes a boundary-complete comparison: the unit-lapse
onset obstruction, its regularized Israel measure, the curvature-only closure,
and the lapse-only positive family.  We do not claim a new general matter
model, a new existence theorem for hollow shells, a stability theorem, or a
transport result.  Material support, formation, nonspherical evolution, and
dynamical transport are separate questions; the present family is a controlled
static source geometry, not a laboratory proposal or a warp drive.

\section*{Statements and Declarations}

\textbf{Funding.} No external funding was received for this work.

\textbf{Competing interests.} The authors declare no competing financial or
non-financial interests directly related to this work.

\section*{Code and data availability}

{\raggedright
\textbf{Code availability.} The symbolic and numerical results reported here
can be reproduced with
the scripts \path{verify_obstruction.py},
\path{verify_pg_israel_covariant_regularization.py},
\path{verify_curvature_only_no_go.py},
\path{verify_lapse_only_escape.py},
\path{verify_profile_family_junction_benchmark.py}, and
\path{verify_vlasov_constitutive_branch.py}.  The corresponding result files are
\path{verification.json},
\path{pg_israel_covariant_regularization.json},
\path{profile_family_junction_benchmark.json}, and
\path{vlasov_constitutive_branch.json}.  These scripts and result files are
provided as Supplemental Material.

\textbf{Data availability.} The machine-readable outputs and grayscale
figures, together with the source files needed to interpret them, are
available in the public archive at
\href{https://doi.org/10.5281/zenodo.21923546}{doi:10.5281/zenodo.21923546}
and mirrored at
\href{https://github.com/AstrumDrive/hollow-core-energy-conditions-reproducibility}
{github.com/AstrumDrive/hollow-core-energy-conditions-reproducibility}.\par}

\bibliographystyle{apsrev4-2}
\bibliography{references}

\end{document}